\documentclass[11pt]{article}

\usepackage[active]{srcltx}
\usepackage{hyperref}
\usepackage{amsmath,amssymb}
\usepackage[all]{xy}

 1 
\newtheorem{theorem}{Theorem}
\newtheorem{proposition}{Proposition}

\newtheorem{definition}{Definition}

\def\qed{\ifvmode\Realemovelastskip\fi
{\unskip\nobreak\hfil\penalty50\hbox{}\nobreak\hfil \hbox{\vrule
height1.2ex width1.2ex}\parfillskip=0pt \finalhyphendemerits=0
\par\smallskip}}
\def\qedr{\ifvmode\Realemovelastskip\fi
{\unskip\nobreak\hfil\penalty50\hbox{}\nobreak\hfil \hbox{
$\diamond$}\parfillskip=0pt \finalhyphendemerits=0
\par\smallskip}}
\def\ds{\displaystyle}
\newenvironment{proof}{\noindent{\sl Proof:~~~}}{\quad \qed}

\def\beq{\begin{equation}}
\def\eeq{\end{equation}}
\def\bea{\begin{eqnarray}}
\def\eea{\end{eqnarray}}
\def\beann{\begin{eqnarray*}}
\def\eeann{\end{eqnarray*}}
\def\beasn{\begin{sneqnarray}}
\def\eeasn{\end{sneqnarray}}
\def\ben{\begin{enumerate}}
\def\een{\end{enumerate}}
\def\bit{\begin{itemize}}
\def\eit{\end{itemize}}

\def\derpar#1#2{\displaystyle\frac{\partial{#1}}{\partial{#2}}}
\def\derpars#1#2#3{\displaystyle\frac{\partial^2{#1}}{\partial{#2}\partial{#3}}}
\def\restric#1#2{\left.#1\right|_{#2}}

\def\W{{\cal W}}
\def\C{{\cal C}}
\def\P{{\cal P}}

\def\vf{{\mathfrak{X}}}
\def\df{{\mit\Omega}}
\def\Lag{{\cal L}}
\def\Leg{{\cal FL}}

\def\d{{\rm d}}

\def\Nat{\mathbb{N}}

\def\Real{\mathbb{R}}
\def\R{\mathbb{R}}

\def\Tan{{\rm T}}
\def\Lie{\mathop{\rm L}\nolimits}
\def\inn{\mathop{i}\nolimits}
\def\Cinfty{{\rm C}^\infty}

\def\tabaddress#1{{\small\it\begin{tabular}[t]{c}#1
\\[1.2ex]\end{tabular}}}
\def\qed{\ifvmode\removelastskip\fi
{\unskip\nobreak\hfil\penalty50\hbox{}\nobreak\hfil \hbox{\vrule
height1.2ex width1.2ex}\parfillskip=0pt \finalhyphendemerits=0
\par\smallskip}}

\title{HIGHER-ORDER MECHANICS: VARIATIONAL PRINCIPLES AND OTHER TOPICS}

\author{
{\sc  Pedro Daniel Prieto-Mart\'\i nez\thanks{{\bf e}-{\it mail}:
   peredaniel@ma4.upc.edu} }\\
   {\sc Narciso Rom\'an-Roy\thanks{{\bf e}-{\it mail}:
   nrr@ma4.upc.edu}}  \\
   \tabaddress{Departamento de Matem\'atica Aplicada IV.
   Edificio C-3, Campus Norte UPC\\
   C/ Jordi Girona 1. 08034 Barcelona. Spain}
}
   \date{December 18, 2012}

\begin{document}

\maketitle

\pagestyle{myheadings}

\thispagestyle{empty}

\begin{abstract}
After reviewing the Lagrangian-Hamiltonian unified formalism (i.e,  the Skinner-Rusk formalism)
for higher-order (non-autonomous) dynamical systems,
we state a unified geometrical version of the Variational Principles 
which allows us  to derive the
Lagrangian and Hamiltonian equations for these kinds of systems.
Then, the standard Lagrangian and Hamiltonian formulations of these principles and the corresponding
dynamical equations are recovered from this unified framework.
\end{abstract}


\vbox{\raggedleft AMS s.\,c.\,(2010): 35A15 \and 37B55 \and 70H50}\null 
\markright{\rm P.D. Prieto-Mart\'{\i}nez, N. Rom\'{a}n-Roy:   \sl Higher-order Mechanics: Variational Principles...}

\clearpage


\section{Introduction}
\label{intro}

Higher-order systems appear in many models in
theoretical and mathematical physics, such as
in the mathematical description of relativistic particles with spin, string theories,
gravitation, Podolsky's generalization of electromagnetism and others.
They also appear
in some problems of fluid mechanics, electric networks and classical physics, and 
in numerical models arising from the geometric discretization of first-order dynamical systems
that preserve their inherent geometric structures
(see \cite{art:Prieto_Roman11_1,art:Prieto_Roman11_2}
for a long but non-exhaustive list of references).

In these kinds of systems and, in particular, in higher-order mechanics,
the dynamics have explicit dependence on accelerations or higher-order derivatives 
of the generalized coordinates of position.
So, for Lagrangian systems, if the Lagrangian function depends on derivatives of order $k$,
the corresponding Euler-Lagrange equations are of order $2k$.
Thus, the geometric descriptions of these systems use higher-order tangent and jet bundles
as the main tool (see, for instance, 
\cite{proc:Cantrijn_Crampin_Sarlet86,book:DeLeon_Rodrigues85,proc:DeLeon_Rodrigues87,art:Gracia_Pons_Roman91,art:Gracia_Pons_Roman92,art:Krupkova00,book:Saunders89}).

Furthermore, a generalization of the Lagrangian and Hamiltonian formalisms
of first-order autonomous mechanical systems exists  that compresses them
into a single formalism: the {\sl Skinner-Rusk} or {\sl Lagrangian-Hamiltonian unified formalism},
proposed by R. Skinner and R. Rusk 
for first-order autonomous mechanical systems \cite{art:Skinner_Rusk83}. 
It was generalized to non-autonomous dynamical systems,
control systems, first-order classical field theories 
and higher-order classical field theories
\cite{art:Barbero_Echeverria_Martin_Munoz_Roman08,art:Barbero_Echeverria_Martin_Munoz_Roman07,art:Cortes_Martinez_Cantrijn02,art:Echeverria_Lopez_Marin_Munoz_Roman04}.
The generalization of the Skinner-Rusk unified  formalism for higher-order mechanical systems
has been developed in recent papers
\cite{art:Colombo_Martin_Zuccalli10,art:Prieto_Roman11_1,art:Prieto_Roman11_2}.

The aim of this lecture is twofold:
first to review this unified formalism for higher-order mechanical systems and
second to state the variational principles for higher-order systems 
and derive the higher-order Euler-Lagrange and Hamilton equations
using this unified framework.
These geometric variational techniques are based on those introduced for for first-order field theories 
in \cite{art:Aldaya_Azcarraga78,GC-74,proc:Garcia_Munoz83,art:Goldschmidt_Sternberg73}.
Our study is made for non-autonomous higher-order mechanical systems
(the autonomous case can thereby be obtained by using trivial bundles
and removing the time-dependence).

In particular, we start by introducing some basic geometrical background
in Section \ref{geomset}, and then reviewing the construction of the
framework for the unified Lagrangian-Hamiltonian formalism for non-autonomous higher-order systems
(developed in \cite{art:Prieto_Roman11_2}) in Section \ref{laghamunif}.
The main contributions of the paper begin in Section \ref{varprincuf},
where we establish the variational principle and use it to derive
the higher-order equations for the Lagrangian-Hamiltonian unified formalism,
which are written in several equivalent geometric ways.
Then, in Section \ref{audeq}, these equations are analyzed in detail,
showing how they compress not only the dynamical evolution equations but also
the equations of the Legendre-Ostrogradsky transformation connecting the 
Lagrangian and Hamiltonian formalism, which appear as compatibility 
and consistency conditions for the equations.
Other relevant results are presented in Sections \ref{Lagvp} and \ref{Hamvp},
where first we recover the generalization to higher-order systems
of the {\sl Hamilton Variational Principle} of the Lagrangian
formalism and the {\sl Hamilton-Jacobi Variational Principle} of the Hamiltonian
formalism, and then the higher-order Euler-Lagrange and the Hamilton equations.
All these results are obtained in a straightforward way from this unified formalism.
Finally, some conclusions and further research on these topics
are discussed in Section \ref{concl}.

All the manifolds are real, second countable and $\Cinfty$.
The maps and the structures are assumed to be $\Cinfty$. 
Sum over repeated indices is understood.

\section{Higher-order jet bundles over $\R$}
\label{geomset}

(See \cite{book:Saunders89} for details on jet bundles and \cite{art:Prieto_Roman11_2} for details
and proofs on the unified formalism).

Let $E \stackrel{\pi}{\longrightarrow} \R$ be a fiber bundle with $\dim E = n+1$,
and let $\eta \in \df^1(\R)$ be the canonical volume form in $\R$.
If $k\in\Nat$, the \textsl{$k$th order jet bundle} of the projection
$\pi$, $J^{k}\pi$, is the $((k+1)n+1)$-dimensional
manifold of the $k$-jets of sections $\phi \in \Gamma(\pi)$.
A point in $J^{k}\pi$ is denoted by $j^{k}\phi$, where $\phi \in \Gamma(\pi)$
is any representative of the equivalence class.
We have the following natural projections: if $r \leqslant k$,
$$
\begin{array}{rcl}
\pi^k_r \colon J^{k}\pi & \longrightarrow & J^r\pi \\
j^k\phi & \longmapsto & j^r\phi
\end{array} \quad , \quad
\begin{array}{rcl}
\pi^k \colon  J^{k}\pi & \longrightarrow & E \\
j^k\phi & \longmapsto & \phi
\end{array} \quad , \quad
\begin{array}{rcl}
\bar\pi^k= \pi \circ \pi^k \colon  J^{k}\pi & \longrightarrow & \Real \\
j^k\phi & \longmapsto & \pi(\phi)
\end{array} \, .
$$
Notice that $\pi^{k}_{0} = \pi^{k}$, where $J^0\pi$ is canonically identified with
$E$, and $\pi^{k}_{k} = {\rm Id}_{J^{k}\pi}$.
Furthermore, if $\phi \in \Gamma(\pi)$ is a section of $\pi$, we also denote
the canonical lifting of $\phi$ to $J^{k}\pi$ by
$j^{k}\phi \in \Gamma(\bar{\pi}^{k})$.

Let $t$ be the global coordinate in $\R$ such that $\eta = \d t$, and $(t,q^A_0)$,
$1 \leqslant A \leqslant n$,
local coordinates in $E$ adapted to the bundle structure.
Then, natural coordinates in $J^{k}\pi$ are $(t,q_0^A,q_1^A,\ldots,q_k^A)\equiv (t,q_i^A)$, with
$\ds q_0^A = \phi^A , \ q_i^A = \frac{d^i\phi^A}{dt^i}$.
Using these coordinates, the local expressions of the natural projections are
$$
\pi^k_r(t,q_i^A) = (t,q_j^A) \quad , \quad
\pi^k(t,q_i^A) = (t,q_0^A) \quad , \quad
\bar{\pi}^k(t,q_i^A) = (t) \ .
$$

A section $\psi \in \Gamma(\bar{\pi}^k)$ is {\sl holonomic of type $r$},
 $1 \leqslant r \leqslant k$,
if $j^{k-r+1}\phi = \pi^{k}_{k-r+1} \circ \psi$, 
where $\phi = \pi^{k} \circ \psi \in \Gamma(\pi)$;
that is, the section $\psi$ is the lifting of a section of $\pi$ to 
$J^{k-r+1}\pi$.
In particular, a section $\psi$ is {\sl holonomic of type 1}
if $j^{k}\phi = \psi$; that is,
$\psi$ is the canonical $k$-jet lifting of a section $\phi \in \Gamma(\pi)$,
where $\phi = \pi^{k} \circ \psi$.
A vector field $X \in \vf(J^k\pi)$ is a {\sl semispray of type $r$}
if every integral section of $X$ is holonomic of type $r$.
Throughout this paper, sections that are
holonomic of type $1$ are simply called {\sl holonomic}.
$$
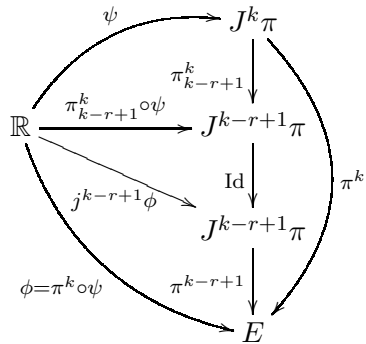
\xymatrix{
\ & \ & J^k\pi \ar[d]_{\pi^k_{k-r+1}} \ar@/^2.5pc/[ddd]^{\pi^k} \\
\R \ar@/^1.5pc/[urr]^{\psi} \ar@/_1.5pc/[ddrr]_{\phi = \pi^k\circ\psi}
\ar[rr]^-{\pi^k_{k-r+1}\circ\psi} \ar[drr]_{j^{k-r+1}\phi} 
& \ & J^{k-r+1}\pi \ar[d]_{{\rm Id}} \\
\ & \ & J^{k-r+1}\pi \ar[d]_{\pi^{k-r+1}} \\
\ & \ & E
}
$$
In natural coordinates, the local expression of a holonomic section of type $r$ is
$$
\psi(t) = (t,q_0^A,q_1^A,\ldots,q_{k-r+1}^A,\psi_{k-r+2}^A,\ldots,\psi_k^A) \ .
$$
Thus, the local expression of a semispray of type $r$ is
$$
X = \derpar{}{t} + q_1^A\derpar{}{q_0^A} + \ldots + q_{k-r+1}^A\derpar{}{q_{k-r}^A}
+ X_{k-r+1}^A\derpar{}{q_{k-r+1}^A} + \ldots + X_k^A\derpar{}{q_k^A} \ .
$$

\section{Lagrangian-Hamiltonian unified formalism}
\label{laghamunif}

Let $\pi\colon E \to\R$ be the configuration bundle
of a $k$th order dynamical system, with $\dim E = n+1$.
The  {\sl higher-order extended jet-momentum bundle} and
the {\sl higher-order restricted jet-momentum bundle} are
$$
\W = J^{2k-1}\pi \times_{J^{k-1}\pi} \Tan^*(J^{k-1}\pi) \quad ; \quad 
\W_r = J^{2k-1}\pi \times_{J^{k-1}\pi} \, J^{k-1}\pi^* \, ,
$$
where $J^{k-1}\pi^* = \Tan^*(J^{k-1}\pi)/(\bar{\pi}^{k-1})^*\Tan^*\R$.

(Observe that $\dim\Tan^*(J^{k-1}\pi) = 2kn+2 > 2kn+1 = \dim J^{2k-1}\pi=\dim J^{k-1}\pi^*$).

These bundles are endowed with the canonical projections
\begin{align*}
\rho_1 \colon \W \to J^{2k-1}\pi \quad &; \quad \rho_2 \colon \W \to \Tan^*(J^{k-1}\pi) \\
\rho_{J^{k-1}\pi} \colon \W \to J^{k-1}\pi \quad &; \quad \rho_\R \colon \W \to \R \\
\rho_1^r \colon \W_r \to J^{2k-1}\pi \quad &; \quad \hat{\rho}_2^r \colon \W_r \to J^{k-1}\pi^* \\
\rho_{J^{k-1}\pi}^r \colon \W_r \to J^{k-1}\pi \quad &; \quad \rho_\R^r \colon \W_r \to \R \ .
\end{align*}
In addition, the natural quotient map  $\mu \colon \Tan^*(J^{k-1}\pi) \to J^{k-1}\pi^*$ 
induces a natural projection $\mu_\W \colon \W \to \W_r$.
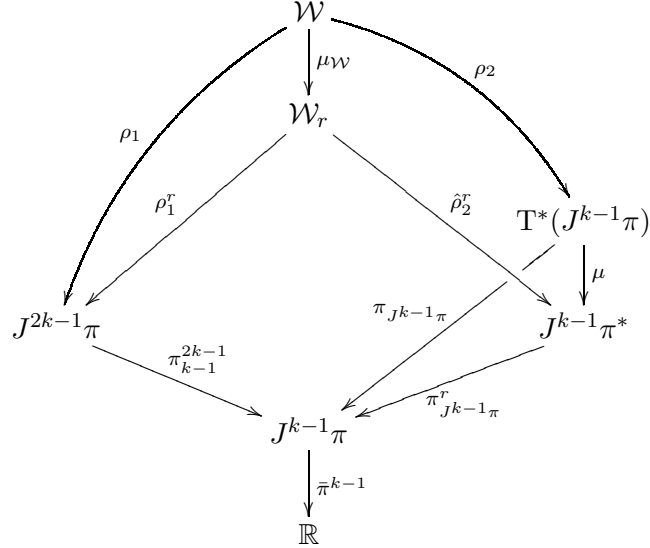
$$
\xymatrix{
\ & \ & \W \ar@/_1.3pc/[llddd]_{\rho_1} \ar[d]^-{\mu_\W} \ar@/^1.3pc/[rrdd]^{\rho_2} & \ & \ \\
\ & \ & \W_r \ar[lldd]_{\rho_1^r} \ar[rrdd]^{\hat{\rho}_2^r} & \ & \ \\
\ & \ & \ & \ & \Tan^*(J^{k-1}\pi) \ar[d]^-{\mu} \ar[lldd]_{\pi_{J^{k-1}\pi}}|(.25){\hole} \\
J^{2k-1}\pi \ar[rrd]^{\pi^{2k-1}_{k-1}} & \ & \ & \ & J^{k-1}\pi^* \ar[dll]^{\pi_{J^{k-1}\pi}^r} \\
\ & \ & J^{k-1}\pi \ar[d]^{\bar{\pi}^{k-1}} & \ & \ \\
\ & \ & \R & \ & \
}
$$

If $(t,q_0^A)$
are local coordinates in $E$ adapted to the bundle structure,
the induced coordinates in all these bundles are
\begin{align*}
&J^{2k-1}\pi \colon (t,q_0^A,\ldots,q_{k-1}^A,q_{k}^A,\ldots,q_{2k-1}^A)
 \equiv (t,q_i^A,q_j^A).\\
&\Tan^*(J^{k-1}\pi) \colon (t,q_0^A,\ldots,q_{k-1}^A,p,p_A^0,\ldots,p_A^{k-1})
\equiv(t,q_i^A,p,p_A^i).\\
&J^{k-1}\pi^* \colon (t,q_0^A,\ldots,q_{k-1}^A,p_A^0,\ldots,p_A^{k-1})
\equiv (t,q_i^A,p^i_A).\\
&\W \colon (t,q_0^A,\ldots,q_{k-1}^A,q_k^A,\ldots,q_{2k-1}^A,p,p_A^0,\ldots,p_A^{k-1})
\equiv (t,q_i^A,q_j^A,p,p_A^i).\\
&\W_r \colon (t,q_0^A,\ldots,q_{k-1}^A,q_k^A,\ldots,q_{2k-1}^A,p_A^0,\ldots,p_A^{k-1})
\equiv(t,q_i^A,q_j^A,p_A^i).
\end{align*}

{\bf Remark:}
The last coordinates are the real dynamical variables and so
$\W_r$ is the real phase space in the unified formalism.

Observe that $\dim\,\W = 3kn + 2$ and $\dim\,\W_r = 3kn + 1$.

\begin{definition}
A section $\psi \in \Gamma(\rho_\R)$ is {\rm holonomic of type $r$} in $\W$,
$1 \leqslant r \leqslant 2k-1$, if the section $\psi_1 = \rho_1 \circ \psi \in \Gamma(\bar{\pi}^{2k-1})$
is holonomic of type $r$ in $J^{2k-1}\pi$.

\noindent A vector field $X \in \vf(\W)$ is a {\rm semispray of type $r$} in $\W$,
$1 \leqslant r \leqslant k$,
if every integral section $\psi$ of $X$ is holonomic of type $r$ in $\W$.
\end{definition}

Let $\Theta_{k-1} \in \df^{1}(\Tan^*(J^{k-1}\pi))$
and $\Omega_{k-1} = -\d\Theta_{k-1} \in \df^{2}(\Tan^*(J^{k-1}\pi))$
be the canonical forms on $\Tan^*(J^{k-1}\pi)$.
The {\sl higher-order unified canonical forms} are
$\Theta = \rho_2^*\Theta_{k-1} \in \df^{1}(\W)$ and
$\Omega = \rho_2^*\Omega_{k-1} \in \df^{2}(\W)$.
Notice that $\ker\Omega = \ker{\rho_2}_*$,
and then $(\W,\Omega,\rho_\R^*\eta)$ is a precosymplectic manifold.

In natural coordinates, the above forms are given by
$$
\Theta = p_A^i\d q_i^A + p\d t \quad , \quad
\Omega = \d q_i^A \wedge \d p_A^i - \d p \wedge \d t \, ,
$$
and $\ker\Omega$ is locally given by
$$
\ker\Omega = \left\langle \derpar{}{q_k^A},\ldots,\derpar{}{q_{2k-1}^A} \right\rangle \ .
$$

\begin{definition}
The {\rm higher-order coupling $1$-form} $\hat{\C}\in\df^{1}(\W)$
is the $\rho_\R$-semibasic form defined as follows:
for every $w = (\bar{y},\alpha_q) \in \W$, where $\bar{y} \in J^{2k-1}\pi$,
$q = \pi^{2k-1}_{k-1}(\bar{y})$, and $\alpha_q \in \Tan_q^*(J^{k-1}\pi)$; 
if $\phi \in \Gamma(\pi)$ is any representative of $\bar{y}$,
and $u \in \Tan_w\W$, then
\begin{equation*}
\langle \hat{\C}(w) \mid u \rangle = \langle \alpha_q \mid (\Tan_w(j^{k-1}\phi \circ \rho_\R))(u) \rangle \, .
\end{equation*}
\end{definition}

As $\hat{\C}$ is a $\rho_\R$-semibasic form, there is
a {\sl coupling function} $\hat{C} \in \Cinfty(\W)$ such that
$\hat{\C} = \hat{C}\rho_\R^*\eta = \hat{C}\d t$. In natural coordinates
the coupling function is
\begin{equation}
\label{eqn:LocalCoordCouplingFunct}
\hat{C} = p + p_A^iq_{i+1}^A \, .
\end{equation}

The dynamical information is introduced by giving a
{\sl $kth$-order Lagrangian density} $\Lag \in \df^{1}(J^{k}\pi)$, which is
a $\bar{\pi}^k$-semibasic form. So we can write $\Lag = L \cdot (\bar{\pi}^k)^*\eta$,
where $L \in \Cinfty(J^{k}\pi)$ is the {\sl Lagrangian function}.
Then we denote
$\hat{\Lag} = (\pi^{2k-1}_{k} \circ \rho_1)^*\Lag$.
As the Lagrangian density is a $\bar{\pi}^{k}$-semibasic form, then
$\hat{\Lag}$ is a $\rho_\R$-semibasic $1$-form, and we have that
$\hat{\Lag} = \hat{L}\rho_\R^*\eta = \hat{L}\d t$, where
$\hat{L} = (\pi^{2k-1}_{k} \circ \rho_1)^*L \in \Cinfty(\W)$. 

In order to have a geometric structure in $\W_r$ we define
the so-called \textsl{Hamiltonian submanifold}
\begin{equation*}
\W_o = \left\{ w \in \W \colon \hat{\Lag}(w) = \hat{\C}(w) \right\} \stackrel{j_o}{\hookrightarrow} \W \, .
\end{equation*}
Since $\hat{\C}$ and $\hat{\Lag}$ are both $\rho_\R$-semibasic forms,
the submanifold $\W_o$ is defined by the constraint $\hat{C} - \hat{L} = 0$.
In natural coordinates, bearing in mind the local expression
\eqref{eqn:LocalCoordCouplingFunct} of $\hat{C}$,
the constraint function is given by
$p + p_A^iq_{i+1}^A - \hat{L} = 0$.

From \cite{art:Prieto_Roman11_2} we have that:

\begin{proposition}
\label{propaux}
The submanifold $\W_o \hookrightarrow \W$ is $1$-codimensional,
$\mu_\W$-transverse and
diffeomorphic to $\W_r$.
\end{proposition}

As a consequence of this, if $\Upsilon\colon\W_r\to\W_o$ denotes this diffeomorphism,
we have an induced section $\hat{h}=j_o\circ\Upsilon \in \Gamma(\mu_\W)$, which is
specified by giving the local \textsl{Hamiltonian function}
$\hat{H} = -\hat{L} + p_A^iq_{i+1}^A \in \Cinfty(\W_r)$;
that is, we have 
$$
\hat{h}(t,q_i^A,q_j^A,p_A^i) = (t,q_i^A,q_j^A,-\hat{H},p_A^i) \equiv
(t,q_i^A,q_j^A,\hat{L}-p_A^iq_{i+1}^A,p_A^i) \ .
$$
The section $\hat{h}$ is
called a \textsl{Hamiltonian section} of $\mu_\W$, or a \textsl{Hamiltonian $\mu_\W$-section}.
$$
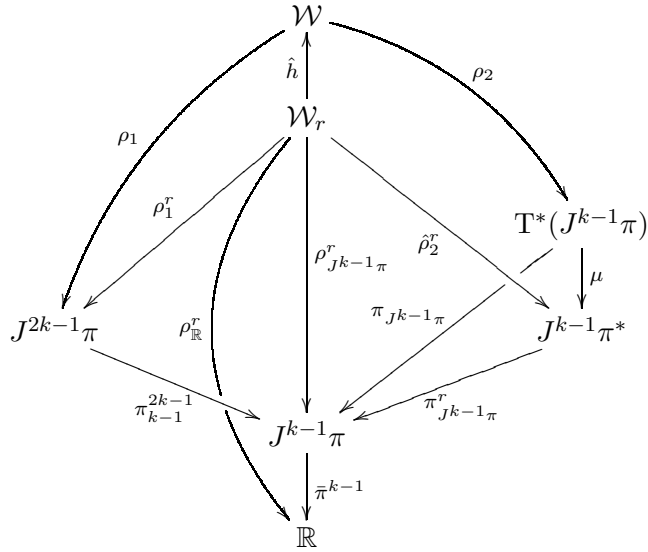
\xymatrix{
\ & \ & \W \ar@/_1.3pc/[llddd]_{\rho_1} \ar@/^1.3pc/[rrdd]^{\rho_2} & \ & \ \\
\ & \ & \W_r \ar[u]^{\hat h} \ar[lldd]_{\rho_1^r} \ar[rrdd]_{\hat{\rho}_2^r} \ar[ddd]^<(0.4){\rho_{J^{k-1}\pi}^r} 
\ar@/_3pc/[dddd]_-{\rho_\R^r}|(.675){\hole} & \ & \ \\
\ & \ & \ & \ & \Tan^*(J^{k-1}\pi) \ar[d]^-{\mu} \ar[lldd]_{\pi_{J^{k-1}\pi}}|(.25){\hole} \\
J^{2k-1}\pi \ar[rrd]_{\pi^{2k-1}_{k-1}} & \ & \ & \ & J^{k-1}\pi^* \ar[dll]^{\pi_{J^{k-1}\pi}^r} \\
\ & \ & J^{k-1}\pi \ar[d]^{\bar{\pi}^{k-1}} & \ & \ \\
\ & \ & \R & \ & \
}
$$

Next, we define the forms
$\Theta_r = \hat{h}^*\Theta \in \df^{1}(\W_r)$
and $\Omega_r = \hat{h}^*\Omega \in \df^{2}(\W_r)$,
whose expressions in natural coordinates are
\begin{equation*}
\Theta_r = p_A^i\d q_i^A + (\hat{L} - p_A^iq_{i+1}^A)\d t \quad ; \quad
\Omega_r = \d q_i^A \wedge \d p_A^i + \d(p_A^iq_{i+1}^A - \hat{L}) \wedge \d t \ .
\end{equation*}

{\bf Remark:}
The precosymplectic Hamiltonian system $(\W_r,\Omega_r,(\rho_\R^r)^*\eta)$
(or $(\W_o,\Omega_o,(\rho_\R^o)^*\eta)$, with $\Omega_o = j_o^*\Omega$)
represents the higher-order dynamical system in the Lagrangian-Hamiltonian unified formalism.

\section{Variational Principle for the unified formalism}
\label{varprincuf}

Next we establish the variational principle from which the dynamical
equations for the unified formalism are derived.
Our starting point is the
precosymplectic Hamiltonian system $(\W_r,\Omega_r,(\rho_\R^r)^*\eta)$.

Let $\Gamma(\rho_\R^r)$ be the set of sections of $\rho_\R^r$, that is,
curves $\psi \colon \R \to \W_r$. Consider the functional
\begin{equation*}
\begin{array}{rcl}
\mathbf{LH} \colon \Gamma(\rho_\R^r) & \longrightarrow & \R \\
\psi & \longmapsto & \displaystyle \int_\R \psi^*\Theta_r
\end{array} \, ,
\end{equation*}
where the convergence of the integral is assumed.

\begin{definition}[Generalized Variational Principle]
The {\rm Lagrangian-Hamiltonian variational problem} for the system
$(\W_r,\Omega_r,(\rho_\R^r)^*\eta)$
is the search for the critical (or stationary) holonomic sections of
the functional $\mathbf{LH}$ with respect to the variations
of $\psi$ given by $\psi_t = \sigma_t \circ \psi$,
where $\left\{ \sigma_t \right\}$ is a local one-parameter group of
any compact-supported vector field $Z \in \vf^{V(\rho_\R^r)}(\W_r)$, that is,
\begin{equation*}
\restric{\frac{d}{dt}}{t=0}\int_\R \psi_t^*\Theta_r = 0 \, .
\end{equation*}
\end{definition}

The main result of the calculus of variations in this context is the following:

\begin{theorem}
\label{thm:EquivalenceVariationalSectionsUnified}
The following assertions on a section $\psi \in \Gamma(\rho_\R^r)$ are equivalent:
\begin{enumerate}
\item $\psi$ is a solution to the Lagrangian-Hamiltonian variational problem.
\item $\psi$ is a holonomic section solution to the equation
\begin{equation*}
\psi^*\inn(Y)\Omega_r = 0 \, , \quad \text{for every }Y \in \vf(\W_r) \, .
\end{equation*}
\item $\psi$ is a holonomic section solution to the equation
\begin{equation*}
\inn(\psi^\prime)(\Omega_r \circ \psi) = 0\ ,
\end{equation*}
where $\psi^\prime \colon \R \to \Tan\W_r$ is the canonical lifting of $\psi$ to $\Tan\W_r$.
\item $\psi$ is an integral curve of a vector field
contained in a class of $\rho_\R^r$-transverse semisprays of type $1$,
$\left\{ X \right\} \subset \vf(\W_r)$,
satisfying the equation
\begin{equation}
\label{vf}
\inn(X)\Omega_r = 0 \, .
\end{equation}
\end{enumerate}
\end{theorem}
\begin{proof}
We prove the equivalence $1\,\Longleftrightarrow\, 2$ following
the patterns taken from \cite{art:Echeverria_DeLeon_Munoz_Roman07}.
For the proof of the other equivalences, see
\cite{art:Prieto_Roman11_2} (Theorem 1).

Let $Z \in \vf^{V(\rho_\R^r)}(\W_r)$ be a compact-supported vector field,
and $V \subset \R$ an open set such that $\partial V$ is a $0$-dimensional
manifold and $\rho_\R^r({\rm supp}(Z)) \subset V$. Then,
\begin{align*}
\restric{\frac{d}{dt}}{t=0} \int_\R \psi^*_t\Theta_r
&= \restric{\frac{d}{dt}}{t=0} \int_V \psi^*_t\Theta_r
= \restric{\frac{d}{dt}}{t=0} \int_V \psi^*\sigma_t^*\Theta_r\\
&= \int_V\psi^*\left( \lim_{t\to0} \frac{\sigma_t^*\Theta_r - \Theta_r}{t} \right)
= \int_V\psi^*\Lie(Z)\Theta_r \\
&= \int_V \psi^*(\inn(Z)\d \Theta_r + \d\inn(Z)\Theta_r)\\
&= \int_V \psi^*(-\inn(Z)\Omega_r + \d\inn(Z)\Theta_r)\\
&= - \int_V \psi^*\inn(Z)\Omega_r + \int_V \d(\psi^*\inn(Z)\Theta_r) \\
&= - \int_V \psi^*\inn(Z)\Omega_r + \int_{\partial V}\psi^*\inn(Z)\Theta_r
= - \int_V\psi^*\inn(Z)\Omega_r \, ,
\end{align*}
as a consequence of Stoke's theorem and the assumptions made
on the supports of the vertical vector fields. Thus,
by the fundamental theorem of the variational calculus,
we conclude
$$
\restric{\frac{d}{dt}}{t=0} \int_\R \psi_t^*\Theta_r = 0 \quad
\Longleftrightarrow \quad
\psi^*\inn(Z)\Omega_r = 0
$$
for every compact-supported $Z \in \vf^{V(\rho_\R^r)}(\W_r)$.
However, since the compact-supported vector fields generate locally
the $\Cinfty(\W_r)$-module of vector fields in $\W_r$,
it follows that the last equality holds for every $\rho_\R^r$-vertical
vector field $Z$ in $\W_r$.

Now, recall that for every point $w \in {\rm Im}\psi$, we have a canonical
splitting of the tangent space of $\W_r$ at $w$ in a $\rho_\R^r$-vertical
subspace and a $\rho_\R^r$-horizontal subspace, that is,
$$
\Tan_w\W_r = V_w(\rho_\R^r) \oplus \Tan_w({\rm Im}\psi) \, .
$$
Thus, if $Y \in \vf(\W_r)$, then
$$
Y_w = (Y_w - \Tan_w(\psi \circ \rho_\R^r)(Y_w)) + 
\Tan_w(\psi \circ \rho_\R^r)(Y_w) \equiv Y_w^V + Y_w^{\psi} \, ,
$$
with $Y_w^V \in V_w(\rho_\R^r)$ and $Y_w^{\psi} \in \Tan_w({\rm Im}\psi)$. Therefore
$$
\psi^*\inn(Y)\Omega_r= \psi^*\inn(Y^V)\Omega_r + \psi^*\inn(Y^{\psi})\Omega_r = 
\psi^*\inn(Y^{\psi})\Omega_r \, ,
$$
since $\psi^*\inn(Y^V)\Omega_r = 0$, by the conclusion in the above
paragraph. Now, as $Y^{\psi}_w \in \Tan_w({\rm Im}\psi)$
for every $w \in {\rm Im}\psi$,
then the vector field $Y^{\psi}$ is tangent to ${\rm Im}\psi$,
and hence there exists a vector field
$X \in \vf(\R)$ such that $X$ is $\psi$-related with $Y^{\psi}$;
that is, $\psi_*X = \restric{Y^{\psi}}{{\rm Im}\psi}$. Then
$\psi^*\inn(Y^{\psi})\Omega_r = \inn(X)\psi^*\Omega_r$. However, as
$\dim{\rm Im}\psi=\dim\R = 1$ and
$\Omega_r$ is a $2$-form, we obtain that $\psi^*\inn(Y^{\psi})\Omega_r = 0$.
Hence, we conclude that $\psi^*\inn(Y)\Omega_r = 0$ for
every $Y \in \vf(\W_r)$.

Taking into account the reasoning of the first paragraph,
the converse is obvious
since the condition $\psi^*\inn(Y)\Omega_r = 0$, for every $Y \in \vf(\W_r)$,
holds, in particular, for every $Z \in \vf^{V(\rho_\R^r)}(\W_r)$.
\end{proof}

\section{Analysis of the unified dynamical equations}
\label{audeq}

In order to complete the Lagrangian-Hamiltonian unified formalism,
it is necessary to analyze the dynamical equations.
We start by using the equations written for vector fields.
Thus, equations \eqref{vf}, with the $\rho_\R^r$-transverse condition, are
$$
\inn(X)\Omega_r = 0 \quad ; \quad
\inn(X)(\rho_\R^r)^*\eta \neq 0 \ .
$$
It is usual to fix the $\rho_\R^r$-transverse condition by demanding that
\begin{equation}
\label{gaugefix}
\inn(X)(\rho_\R^r)^*\eta = 1\ .
\end{equation}
This selects a representative in the class $\{ X\}$.
We will do this in the sequel.

The first important result is  \cite{art:Prieto_Roman11_2}:

\begin{proposition}
The above equations are compatible only on the points of the following submanifold of $\W_r$
\begin{align*}
\W_1 &= \left\{ w \in \W_r \colon (\inn(Z)\d\hat{H})(w) = 0, \mbox{ for every } Z \in \ker\Omega_r \right\} \\
&= \left\{ w \in \W_r \colon (\inn(Y)\Omega_r)(w) = 0, \text{ for every }Y \in \vf^{V(\hat{\rho}_2^r)}(\W_r) \right\} \ .
\end{align*}
\end{proposition}

In natural coordinates, a generic vector field $X \in \vf(\W_r)$ is given by
$$
X = f\derpar{}{t} + f_i^A\derpar{}{q_i^A} + F_j^A\derpar{}{q_j^A} + G_A^i\derpar{}{p_A^i} \ ,
$$
and the $\rho_\R^r$-transverse condition implies $f\not=0$, and
in particular, using \eqref{gaugefix}, we get $f =1$.
Therefore, the dynamical equation \eqref{vf} first gives
$$
p_A^{k-1} - \derpar{\hat{L}}{q_k^A}= 0 \ ,
$$
which are the compatibility relations (constraints) defining locally $\W_1$.
Furthermore, for $0 \leqslant l \leqslant k-1$, $k \leqslant j \leqslant 2k-1$;
$$
f_i^A = q_{i+1}^A \quad ; \quad
G_A^0 = \derpar{\hat{L}}{q_0^A} \quad ; \quad G_A^{i} =  \derpar{\hat{L}}{q_{i}^A} - p_A^{i-1} \ ;
$$
therefore
$$
X = \derpar{}{t} + q_{i+1}^A\derpar{}{q_i^A} + F_j^A\derpar{}{q_j^A} + 
 \derpar{\hat{L}}{q_0^A}\derpar{}{p_A^0}+
\left( \derpar{\hat{L}}{q_{i}^A} - p_A^{i-1} \right)\derpar{}{p_A^{i}} \ .
$$
Observe that, in a natural way, $X$ is a semispray of type $k$. Nevertheless, the variational principle
requires that $X$ must be a semispray of type $1$, thus
$$
X = \derpar{}{t} + \sum_{l=0}^{2k-2}q_{l+1}^A\derpar{}{q_{l}^A} + F_{2k-1}^A\derpar{}{q_{2k-1}^A}+ 
 \derpar{\hat{L}}{q_0^A}\derpar{}{p_A^0}+\left( \derpar{\hat{L}}{q_i^A} - p_A^{i-1} \right)
\derpar{}{p_A^i}\ .
$$

Next we must require $X$ to be tangent to $\W_1$.
Thus, it is necessary to impose that $\restric{\Lie(X)\xi}{\W_1} = 0$, for every constraint function
$\xi$ defining $\W_1$:
\begin{align*}
\restric{X\left( p_A^{k-1} - \derpar{\hat L}{q_k^A} \right)}{\W_1} = 0 & \quad \Longleftrightarrow \quad
p_A^{k-2}-  \left( \derpar{\hat L}{q_{k-1}^A} - d_T\left( \derpar{\hat L}{q_k^A} \right) \right) = 0
\quad \mbox{(on $\W_1$)} \ .
\end{align*}
(where $d_T = \derpar{}{t} + q_{i+1}^A\derpar{}{q_i^A}$).
Repeating this procedure ($k-1$ steps), we get
\begin{align*}
p_A^0 &- \sum_{i=0}^{k-1}(-1)^i d_T^i\left(\derpar{\hat{L}}{q_{1+i}^A}\right) = 0
\quad \mbox{(on $\W_{k-1}$)} \ .
\end{align*}

Thus we obtain a sequence of submanifolds
(which can also be obtained by applying any other constraint algorithm
\cite{LMMMR-2002,art:Gotay_Nester_Hinds78}),
$$
\W_0 \hookleftarrow \W_1  \hookleftarrow  \ldots  \hookleftarrow \W_k\equiv\W_\Lag \ .
$$

As a consequence of the last equalities we conclude that

\begin{proposition}
The submanifold 
$\W_\Lag$ is the graph of a map $\Leg \colon J^{2k-1}\pi \to J^{k-1}\pi^*$ locally defined by
$$
\Leg^*t = t \quad , \quad
\Leg^*q_{r-1}^A = q_{r-1}^A \quad , \quad
\Leg^*p_A^{r-1} = \sum_{i=0}^{k-r}(-1)^i d_T^i\left( \derpar{\hat{L}}{q_{r+i}^A} \right) \, .
$$
\end{proposition}

\begin{definition}
The map $\Leg \colon J^{2k-1}\pi \to J^{k-1}\pi^*$ is the
{\rm (restricted) Legendre-Ostrogradsky map}.

\noindent A Lagrangian density $\Lag$ is {\rm regular}
if the map $\Leg$ is a local diffeomorphism. 
Otherwise, $\Lag$ is said to be a {\rm singular Lagrangian}.
If $\Leg$ is a global diffeomorphism, then $\Lag$ is said to be {\rm hyperregular}.
\end{definition}

In natural coordinates, the regularity condition for $\Lag$ is
equivalent to
\begin{equation}
\label{reg}
\det\left( \derpars{L}{q_k^B}{q_k^A} \right)(\bar{y}) \neq 0 \, , \ \text{for every }\bar{y} \in J^{k}\pi \, .
\end{equation}

Observe that $X$ is not necessarily tangent to $\W_\Lag$.
Thus, imposing the tangency condition to the last generation of constraints
defining $\W_\Lag$, these conditions give the following equations (on $\W_\Lag$):
$$
(-1)^k\left(F_{2k-1}^B - d_T\left(q_{2k-1}^B\right)\right) \derpars{\hat{L}}{q_k^B}{q_k^A} +
 \sum_{i=0}^{k} (-1)^id_T^i\left( \derpar{\hat{L}}{q_i^A} \right) = 0 \ .
$$
And, as a consequence of \eqref{reg}, we have:

\begin{proposition}
If $\Lag$ is regular, then there is a unique
semispray of type $1$, $X \in \vf(\W_r)$, tangent to $\W_\Lag$, which is a
solution to the dynamical equations (on $\W_\Lag$).
\end{proposition}

If $\Lag$ is not regular, new constraints could appear and the algorithm continues until arriving
(in the best cases) at a {\sl final constraint submanifold} $\W_f\hookrightarrow\W_\Lag$.

If $\psi(t) = (t,q_i^A(t),q_j^A(t),p_A^i(t))$ is an integral section of $X$,
the above equations lead to
$$
\dot{q}_l^A = q_{l+1}^A \quad ; \quad
\dot{q}_{2k-1}^A = F_{2k-1}^A\circ\psi \quad ; \quad
\dot{p}_A^0 = \derpar{\hat{L}}{q_0^A} \quad ; \quad
\dot{p}_A^i = \derpar{\hat{L}}{q_i^A} - p_A^{i-1} \ ,
$$
and after some calculations we reach (on the points of $\W_\Lag$, or $\W_f$),
$$
\derpar{L}{q_0^A}\circ\psi-\frac{d}{dt}\left(\derpar{L}{q_1^A}\circ\psi\right)
+\ldots+ (-1)^k \frac{d^k}{dt^k}\left(\derpar{L}{q_k^A}\circ\psi\right) = 0
$$
$$
\dot{p}_A^0 = \derpar{\hat{L}}{q_0^A} \quad ; \quad
\dot{p}_A^i = \derpar{\hat{L}}{q_i^A} - p_A^{i-1} \ ,
$$
These equations compress both the higher-order Euler-Lagrange
and Hamilton equations, as can be seen in the following sections.

{\bf Remark:}
It is interesting to point out that the variational principle
for higher-order autonomous dynamical systems, and the corresponding dynamical equations,
can be obtained as a particular case of these results when the Lagrangian function
does not depend explicitly on the coordinate $t$.

\section{Lagrangian formalism: Generalized Hamilton Principle}
\label{Lagvp}

In this section we show how to recover the Lagrangian formalism for
higher-order mechanical systems.
In particular, we state the classical Hamilton Variational Principle
of the Lagrangian formalism for higher-order systems
and study its relation with the unified variational Principle.

First, consider the diagram
$$
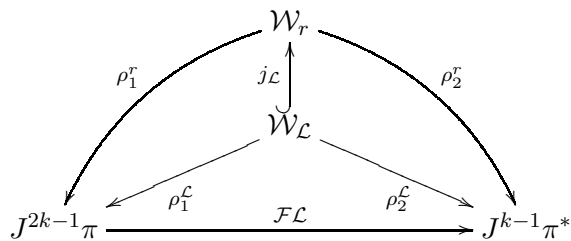
\xymatrix{
\ & \ & \ & \W_r \ar@/_1.3pc/[ddll]_{\rho_1^r} \ar@/^1.3pc/[ddrr]^{\rho_2^r} & \ & \ \\
\ & \ & \ & \W_\Lag \ar[dll]^{\rho_1^\Lag} \ar[drr]_{\rho_2^\Lag} \ar@{^{(}->}[u]^{j_\Lag} & \ & \ \\
\ & J^{2k-1}\pi \ar[rrrr]^{\Leg} & \ & \ & \ & J^{k-1}\pi^* }
$$

As $\W_\Lag$ is the graph of the restricted Legendre-Ostrogadski map,
we have that the map $\rho_1^\Lag = \rho_1^r \circ j_\Lag \colon \W_\Lag \to J^{2k-1}\pi$ is a diffeomorphism.
Then we can define the {\sl Poincar\'{e}-Cartan $1$ and $2$ forms} in $J^{2k-1}\pi$ as
$$
\Theta_{\Lag} = (j_\Lag\circ(\rho^\Lag_1)^{-1})^* \Theta_r\quad ; \quad
\Omega_{\Lag} = - \d\Theta_{\Lag} =( j_\Lag\circ(\rho^\Lag_1)^{-1})^* \Omega_r\ .
$$
These forms can also be introduced in several equivalent ways
(see, for instance,
\cite{art:Aldaya_Azcarraga78,proc:Garcia_Munoz83,art:Saunders87,book:Saunders89}).

{\bf Remark:}
The triple $(J^{2k-1}\pi,\Omega_{\Lag},(\bar{\pi}^{2k-1})^*\eta)$ is the
{\sl higher-order non-autonomous Lagrangian system} associated to $(\W_r,\Omega_r,(\rho_\R^r)^*\eta)$.

Now we establish the variational principle from which we can obtain
the dynamical equations for the Lagrangian formalism.

Given the Lagrangian system $(J^{2k-1}\pi,\Omega_\Lag,(\bar{\pi}^{2k-1})^*\eta)$,
let $\Gamma(\pi)$ be the set of sections of $\pi$, that is, curves $\phi \colon \R \to E$.
Consider the functional
\begin{equation*}
\label{eqn:DefnFunctionalVariationalLagrangian}
\begin{array}{rcl}
\mathbf{L} \colon \Gamma(\pi) & \longrightarrow & \R \\
\phi & \longmapsto & \displaystyle \int_\R (j^{2k-1}\phi)^*\Theta_\Lag
\end{array} \, ,
\end{equation*}
where the convergence of the integral is assumed.

\begin{definition}[Generalized Hamilton Variational Principle]
The {\rm Lagrangian variational problem} 
(also called {\rm Hamilton variational problem})
for the higher-order Lagrangian system
$(J^{2k-1}\pi,\Omega_\Lag,(\bar{\pi}^{2k-1})^*\eta)$
is the search for the critical (or stationary) sections of
the functional $\mathbf{L}$ with respect to the variations
of $\phi$ given by $\phi_t = \sigma_t \circ \phi$,
where $\left\{ \sigma_t \right\}$ is a local one-parameter group of
any compact-supported $Z \in \vf^{V(\pi)}(E)$; that is,
\begin{equation*}
\restric{\frac{d}{dt}}{t=0}\int_\R (j^{2k-1}\phi_t)^*\Theta_\Lag = 0 \, .
\end{equation*}
\end{definition}

Then, as in the above section, we have:

\begin{theorem}
The following assertions on a section $\phi \in \Gamma(\pi)$ are equivalent:
\begin{enumerate}
\item $\phi$ is a solution to the Lagrangian variational problem.
\item $\psi_\Lag =j^{2k-1}\phi$ is a solution to the equation
$$
\psi_\Lag^*\inn(Y)\Omega_\Lag = 0, \quad \mbox{for every }Y \in \vf(J^{2k-1}\pi) \, .
$$
\item $\psi_\Lag = j^{2k-1}\phi$ is a solution to the equation
$$
\inn(\psi^\prime_\Lag)(\Omega_\Lag \circ \psi_\Lag) = 0 \ ,
$$
where $\psi^\prime_\Lag \colon \R \to \Tan J^{2k-1}\pi$ is the canonical lifting of 
$\psi_\Lag$ to $\Tan J^{2k-1}\pi$.
\item $\psi_\Lag =j^{2k-1}\phi$ is an integral curve of a vector field
contained in a class of $\bar{\pi}^{2k-1}$-transverse semisprays of type $1$,
$\left\{ X_\Lag \right\} \subset \vf(J^{2k-1}\pi)$,
satisfying
$$\inn(X_\Lag)\Omega_\Lag = 0 \, .$$
\end{enumerate}
\end{theorem}
\begin{proof}
The proof of the equivalence $1\,\Longleftrightarrow\, 2$
follows the same patterns as in Theorem
\ref{thm:EquivalenceVariationalSectionsUnified}.
For the proof of the other equivalences, see
\cite{art:Prieto_Roman11_2} (Theorem 3).
\end{proof}

Section solutions to the Hamilton variational problem
are recovered from section solutions to the Lagrangian-Hamiltonian
variational problem in the unified formalism. In fact:

\begin{theorem}
\label{prop:LagrangianVariational}
Let $\psi \in \Gamma(\rho_\R^r)$ be a holonomic section which
is a solution to the Lagrangian-Hamiltonian variational problem
given by the functional  \textbf{LH}.
Then, the section
$\psi_\Lag = \rho_1^r \circ \psi \in \Gamma(\bar{\pi}^{2k-1})$ is
holonomic, and its projection $\phi = \pi^{2k-1} \circ \psi_\Lag \in \Gamma(\pi)$
is a solution to the Lagrangian variational problem given by the functional $\mathbf{L}$;

\noindent Conversely, from a holonomic section $\psi_\Lag=j^{2k-1}\phi \in \Gamma(\bar{\pi}^{2k-1})$
which is a solution to the Lagrangian variational problem,
we recover a solution  $\psi=(\psi_\Lag,\psi_\Lag\circ\Leg)$
to the Lagrangian-Hamiltonian variational problem.
$$
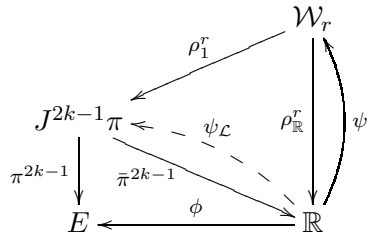
\xymatrix{
\ & \ & \W_r \ar[dd]_-{\rho_\R^r} \ar[dll]_-{\rho_1^r}  \\
J^{2k-1}\pi \ar[d]_{\pi^{2k-1}} \ar[drr]_<(0.25){\bar{\pi}^{2k-1}} & \ & \ \\
E & \ & \R \ar@/_1pc/[uu]_{\psi} \ar@/_1pc/@{-->}[ull]_{\psi_\Lag} \ar[ll]_{\phi} \\
}
$$
\end{theorem}
\begin{proof}
As $\psi \in \Gamma(\rho_\R^r)$ is holonomic,
then $\psi_\Lag = \rho_1^r \circ \psi\in \Gamma(\bar{\pi}^{2k-1})$
is a holonomic section, by definition.

Now, $\rho_1^r$ being a submersion,
for every compact-supported vector field
$X \in \vf^{V(\bar{\pi}^{2k-1})}(J^{2k-1}\pi)$
there exist compact-supported vector fields
$Y \in \vf^{V(\rho_\R^r)}(\W_r)$ such that
${\rho_1^r}_*Y = X$; that is, $X$ and $Y$ are $\rho_1^r$-related.
In particular, this holds if $X$ is the $(2k-1)$-jet lifting
of a compact-supported $\pi$-vertical vector field in $E$; that is, if we
have $X = j^{2k-1}Z$, with $Z \in \vf^{V(\pi)}(E)$.
We denote by $\left\{ \sigma_t \right\}$ a local one-parameter
group for the compact-supported vector fields $Y \in \vf^{V(\rho_\R^o)}(\W_r)$.
Then, using this particular choice of $\rho_1^r$-related vector fields, we have
\begin{align*}
\restric{\frac{d}{dt}}{t=0}\int_\R (j^{2k-1}\phi_t)^*\Theta_\Lag
&= \restric{\frac{d}{dt}}{t=0}\int_\R(j^{2k-1}(\sigma_t \circ\phi))^*\Theta_\Lag \\
&= \restric{\frac{d}{dt}}{t=0}\int_\R(j^{2k-1}\phi)^*(j^{2k-1}\sigma_t)^*\Theta_\Lag \\
&= \int_\R\psi_\Lag^*\Lie(j^{2k-1}Z)\Theta_\Lag 
= \int_\R \psi^*(\rho_1^r)^*\Lie(j^{2k-1}Z)\Theta_\Lag \\
&= \int_\R \psi^*\Lie(Y)\Theta_r 
= \restric{\frac{d}{dt}}{t=0}\int_\R\psi^*\sigma_t^*\Theta_r \\
&= \restric{\frac{d}{dt}}{t=0} \int_\R \psi_t^*\Theta_r = 0 \, ,
\end{align*}
since $\psi$ is a critical section for the Lagrangian-Hamiltonian
variational problem.

Conversely, if we have a holonomic section $\psi_\Lag=j^{2k-1}\phi$ which is a
solution to the Lagrangian variational problem, then we can construct
$\psi=(\psi_\Lag,\psi_\Lag\circ\Leg)$, which is a section
$\psi\colon\R\to\W_\Lag\subset\W_r$ of the projection $\rho_\R^r$
(remember that, in the unified formalism, the dynamical equations
have solutions only on the points of $\W_\Lag$, or in a subset of it).
Then, the above reasoning also shows also that if $\psi_\Lag$
is a solution to the Lagrangian variational problem,
then $\psi$ is a solution to the Lagrangian-Hamiltonian variational problem.
\end{proof}

In natural coordinates, if $\psi$ is given by $\psi(t) = (t,q_i^A(t),q_j^A(t),p_A^i(t))$,
then $\psi_\Lag = (\rho_1^r \circ \psi)(t) = (t,q_i^A(t),q_j^A(t))$,
and $\phi(t) = (\pi^{2k-1} \circ \psi_\Lag)(t) = (t,q_0^A(t))$
satisfies the {\sl $k$th-order Euler-Lagrange equations}
$$
\restric{\derpar{L}{q_0^A}}{j^{2k-1}\phi} - \restric{\frac{d}{dt}\derpar{L}{q_1^A}}{j^{2k-1}\phi}
+ \ldots +
(-1)^k \restric{\frac{d^k}{dt^k}\derpar{L}{q_k^A}}{j^{2k-1}\phi} = 0 \, .
$$

Finally, as a consequence of all the above results, we have the
corresponding relation between vector field solutions to the unified dynamical equations
and those which are solutions to the Lagrangian equations:

\begin{proposition}
Let $X \in \vf(\W_r)$ be a vector field tangent to $\W_\Lag$ 
which is a solution to the equations
\begin{equation}
\label{uno}
\inn(X)\Omega_r = 0 \quad ; \quad \inn(X)(\rho_\R^r)^*\eta = 1 \ .
\end{equation}
Then there exists a unique semispray of type $k$, $X_\Lag \in \vf(J^{2k-1}\pi)$,
which is a solution to the equations
\begin{equation}
\label{dos}
\inn(X_\Lag)\Omega_{\Lag} = 0 \quad ; \quad
\inn(X_\Lag)(\bar{\pi}^{2k-1})^*\eta = 1 \ .
\end{equation}
In addition, if $\Lag$ is a regular
Lagrangian density, then $X_\Lag$ is a semispray of type $1$.

\noindent Conversely, if $X_\Lag \in \vf(J^{2k-1}\pi)$ is a semispray of type $k$
(resp., of type $1$), which is a solution to the equations (\ref{dos}),
then there exists a unique $X \in \vf(\W_r)$ which
is a solution to the equations (\ref{uno})
and it is a semispray of type $k$ in $\W_r$ (resp., of type $1$).
\end{proposition}
\begin{proof}
See also \cite{art:Prieto_Roman11_2} (Theorem 2) for a detailed proof of this statement.
\end{proof}

\section{Hamiltonian formalism: Generalized Hamilton-Jacobi Principle}
\label{Hamvp}

In this section we state the Hamiltonian variational
problem ({\sl Hamilton-Jacobi Principle\/}) for higher-order systems,
recovering it from the unified formalism.
(See \cite{art:Prieto_Roman11_2} for the proofs and details
on the higher-order Hamiltonian formalism).

Consider the restricted Legendre-Ostrogradsky map
$\Leg \colon J^{2k-1}\pi \to J^{k-1}\pi^*$.
First, it can be proved that the following statements are equivalent:
\begin{enumerate}
\item $\Omega_\Lag$ has maximal rank on $J^{2k-1}\pi$.
\item $\Leg \colon J^{2k-1}\pi \to J^{k-1}\pi^*$ is a local diffeomorphism.
\item In natural coordinates of $J^{k}\pi$, $\det\left( \derpars{L}{q_k^B}{q_k^A} \right)(\bar{y}) \neq 0$,
for every $\bar{y} \in J^{k}\pi$.
\end{enumerate}
As stated in Section \ref{audeq},
if these conditions are fulfilled, the Lagrangan density $\Lag$
is said to be {\sl regular}, and when the restricted Legendre-Ostrogradsky
map is a global diffeomorphism, then  $\Lag$ is {\sl hyperregular}.

Now, let
$\widetilde{\P} = {\rm Im}(\widetilde{\Leg}) \stackrel{\tilde{\jmath}}{\hookrightarrow}\Tan^*(J^{k-1}\pi)$
and $\P = {\rm Im}(\Leg) \stackrel{\jmath}{\hookrightarrow}J^{k-1}\pi^*$.
If $\bar{\tau} = \pi_{J^{k-1}\pi}^r \circ \bar{\pi}^{k-1} \colon J^{k-1}\pi^* \to \R$
is the natural projection, we denote
$\bar{\tau}_o = \bar{\tau} \circ \jmath \colon \P \to \R$.
A Lagrangian density $\Lag \in \df^{1}(J^{k}\pi)$ is said to be
{\sl almost-regular} if:
\begin{enumerate}
\item $\P$ is a closed submanifold of $J^{k-1}\pi^*$.
\item $\Leg$ is a submersion onto its image.
\item For every $\bar{y} \in J^{2k-1}\pi$, the fibers $\Leg^{-1}(\Leg(\bar{y}))$ are connected
submanifolds of $J^{2k-1}\pi$.
\end{enumerate}

The Hamiltonian section $\hat{h} \in \Gamma(\mu_\W)$
(introduced after Proposition \ref{propaux})
induces a Hamiltonian section $h \in \Gamma(\mu)$ defined by
$$
\rho_2\circ \hat{h}=h\circ\rho^r_2
$$
Then, if $\Theta_{k-1}$ and $\Omega_{k-1}$ are the canonical $1$ and $2$ forms
of the cotangent bundle $\Tan^*(J^{k-1}\pi)$, we can construct the
{\sl Hamilton-Cartan forms} in $J^{k-1}\pi^*$ and $\P$ by making
\begin{align*}
\Theta_h = h^*\Theta _{k-1} \in \df^{1}(J^{k-1}\pi^*)  \quad &; \quad
\Omega_h = h^*\Omega  \in \df^{2}(J^{k-1}\pi^*) \\
\Theta_{\mathcal P} = \jmath^*\Theta _h \in \df^{1}({\mathcal P})  \quad &; \quad
\Omega_{\mathcal P} = \jmath^*\Omega_h  \in \df^{2}({\mathcal P}) \ .
\end{align*}
Observe that
$\Leg^*\Theta_h = \Theta_\Lag$ and $\Leg^*\Omega_h = \Omega_\Lag$.
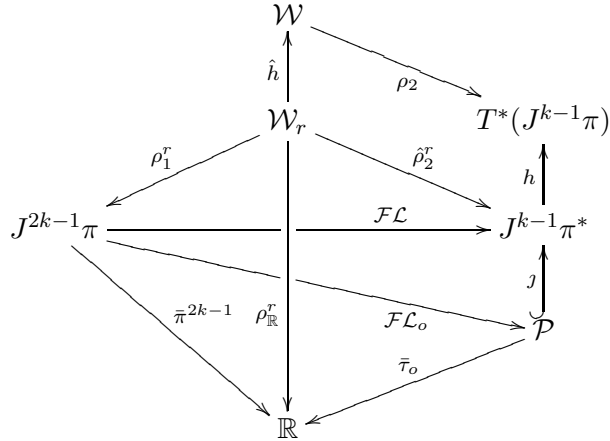
$$
\xymatrix{
\ & \ & \W  \ar[rrd]_{\rho_2} & \ & \ \\
\ & \ & \W_r  \ar[u]^{\hat h} \ar[ddd]_<(0.6){\rho_\R^r} \ar[drr]^{\hat{\rho}_2^r} \ar[dll]_{\rho_1^r} 
& \ & T^*(J^{k-1}\pi)\\
J^{2k-1}\pi \ar[ddrr]^{\bar{\pi}^{2k-1}}
 \ar[rrrr]^<(0.65){\Leg}|(.478){\hole} \ar[drrrr]_<(0.7){\Leg_o}|(.478){\hole} & \ & \ & \ & J^{k-1}\pi^* \ar[u]^{h} \\
\ & \ & \ & \ & \mathcal{P} \ar@{^{(}->}[u]^{\jmath} \ar[dll]_{\bar{\tau}_o} \\
\ & \ & \R & \ & \ \\
}
$$

{\bf Remark:}
$(\P,\Omega_{\mathcal P},\bar{\tau}_o^*\eta)$ is the {\sl higher-order non-autonomous
Hamiltonian system} associated with $(\W_r,\Omega_r,(\rho_\R^r)^*\eta)$.

In what follows, we consider that the Lagrangian density
$\Lag \in \df^{1}(J^{k}\pi)$ is, at least, almost-regular. However,
all the following results also hold for
regular or hyperregular Lagrangian densities,
replacing $\P$ by the corresponding open subset of $J^{k-1}\pi^*$,
or by $J^{k-1}\pi^*$, respectively.

First, we establish the variational principle from which we can obtain the dynamical equations
for the Hamiltonian formalism, and then we show how to obtain the
Hamiltonian dynamical equations.

Given the Hamiltonian system $(\P,\Omega_{\P},\bar{\tau}_o^*\eta)$,
let $\Gamma(\bar{\tau}_o)$ be the set of sections of $\bar{\tau}_o$,
that is, curves $\varphi \colon \R \to \P$. Consider the functional
\begin{equation*}
\label{eqn:DefnFunctionalVariationalHamiltonian}
\begin{array}{rcl}
\mathbf{H} \colon \Gamma(\bar{\tau}_o) & \longrightarrow & \R \\
\varphi & \longmapsto & \displaystyle \int_\R \varphi^*\Theta_{\mathcal P}
\end{array} \, ,
\end{equation*}
where the convergence of the integral is assumed.

\begin{definition}[Generalized Hamilton-Jacobi Variational Principle]
The {\rm Hamiltonian} or {\rm Hamilton-Jacobi variational problem} for
the higher-order Hamiltonian system $(\P,\Omega_{\P},\bar{\tau}_o^*\eta)$
is the search for the critical (or stationary) sections of
the functional $\mathbf{H}$ with respect to the variations
of $\varphi$ given by $\varphi_t = \sigma_t \circ \varphi$,
where $\left\{ \sigma_t \right\}$ is a local one-parameter group of
any compact-supported $Z \in \vf^{V(\bar{\tau}_o)}(\P)$; that is
\begin{equation}
\label{eqn:DynEqVarHamiltonian}
\restric{\frac{d}{dt}}{t=0}\int_\R \varphi_t^*\Theta_{\P} = 0
\end{equation}
\end{definition}

Then, as in the above sections, we have:

\begin{theorem}
The following assertions on a section $\varphi \in \Gamma(\bar{\tau}_o)$ are equivalent:
\begin{enumerate}
\item $\varphi$ is a solution to the Hamiltonian variational problem.
\item $\varphi$ is a solution to the equation
$$\varphi^*\inn(Y)\Omega_{\mathcal P} = 0 \, , \quad \text{for every }Y \in \vf(\P) \, .$$
\item $\varphi$ is a solution to the equation
$$
\inn(\varphi^\prime)(\Omega_{\mathcal P} \circ\varphi) = 0\ ,
$$
where $\varphi^\prime \colon \R \to \Tan\P$ is the canonical lifting of 
$\varphi$ to $\Tan\P$.

\item $\varphi$ is an integral curve of a vector field
contained in a class of $\bar{\tau}_o$-transverse vector fields,
$\left\{ X_h \right\} \subset \vf(\P)$,
satisfying
$$
\inn(X_h)\Omega_{\mathcal P} = 0 \, .
$$
\end{enumerate}
\end{theorem}
\begin{proof}
The proof of the equivalence $1\,\Longleftrightarrow\, 2$
follows the same patterns as in Theorem
\ref{thm:EquivalenceVariationalSectionsUnified}.
For the proof of the other equivalences, see
\cite{art:Prieto_Roman11_2} (Theorem 5).
\end{proof}

In addition, section solutions to the Hamilton equations
are recovered from the solutions to the dynamical equations
in the unified formalism. In fact:

\begin{theorem}
\label{prop:HamiltonianVariational}
Let $\psi \in \Gamma(\rho_\R^r)$ be a
critical section for the Lagrangian-Hamiltonian variational
problem given by the functional  \textbf{LH}. Then, the section
$\psi_h = \Leg_o \circ \rho_1^r \circ \psi = \Leg_o \circ \psi_\Lag \in \Gamma(\bar{\tau}_o)$
is a critical section for the Hamiltonian variational problem
given by the functional  \textbf{H}.

\noindent Conversely, from a section $\psi_h$ solution to the Hamiltonian variational problem,
we recover a solution  $\psi$ to the Lagrangian-Hamiltonian variational problem.
$$
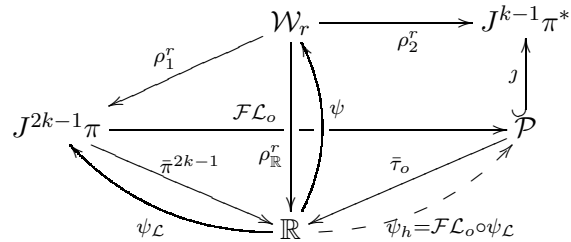
\xymatrix{
\ & \ & \W_r \ar[dd]_<(0.6){\rho_\R^r}  \ar[dll]_{\rho_1^r}  \ar[rr]_{\rho_2^r} & 
\ & J^{k-1}\pi^*& \\
J^{2k-1}\pi \ar[drr]^{\bar{\pi}^{2k-1}} \ar[rrrr]^<(0.35){\Leg_o}|(.5){\hole}|(.565){\hole} & \ & \ & \ 
& \P \ar[dll]_{\bar{\tau}_o}  \ar@{^{(}->}[u]^{\jmath} \\
\ & \ & \R \ar@/_1pc/[uu]_<(0.535){\psi} \ar@/^1pc/[ull]^{\psi_\Lag} 
\ar@/_1pc/@{-->}[urr]_-{\psi_h = \Leg_o \circ \psi_\Lag} & \ & \ 
}
$$
\end{theorem}
\begin{proof}
Observe that $\Leg_o \circ \rho_1^r$ is a submersion,
since it is a composition of submersions, and
$(\Leg_o \circ \rho_1^r)^*\Theta_{\P} = (\rho_1^r)^*(\Leg_o^*\Theta_h) = 
(\rho_1^r)^*\Theta_\Lag = \Theta_r$.
Then, for every compact-supported vector field $Z \in \vf^{V(\bar{\tau}_o)}(\P)$,
there exist compact-supported vector fields $Y \in \vf^{V(\rho_\R^r)}(\W_r)$
such that $(\Leg_o \circ \rho_1^r)_*Y = Z$; that is,
$Z$ is $(\Leg_o \circ \rho_1^r)$-related with $Y$. We denote
by $\{\sigma_t^r\}$ a local one-parameter group for the
compact-supported vector fields $Y \in \vf^{V(\rho_\R^r)}(\W_r)$.
Then, using this particular choice of $(\Leg_o \circ \rho_1^r)$-related vector fields, we have
\begin{align*}
\restric{\frac{d}{dt}}{t=0} \int_\R(\psi_h)_t^*\Theta_{\P}
&= \restric{\frac{d}{dt}}{t=0} \int_\R(\sigma_t \circ \psi_h)^*\Theta_{\P}
= \restric{\frac{d}{dt}}{t=0} \int_\R\psi_h^*(\sigma_t^*\Theta_{\P}) \\
&= \int_\R \psi_h^*\Lie(Z)\Theta_{\P}
= \int_\R \psi^*(\Leg_o \circ \rho_1^r)^*\Lie(Z)\Theta_{\P} \\
&= \int_\R\psi^*\Lie(Y)\Theta_r
= \restric{\frac{d}{dt}}{t=0} \int_\R \psi^*(\sigma_t^r)^*\Theta_r \\
&= \restric{\frac{d}{dt}}{t=0} \int_\R \psi_t^*\Theta_r = 0 \, ,
\end{align*}
since $\psi$ is a critical section for the Lagrangian-Hamiltonian variational problem.

For the converse,
following the same patterns as in the theory of singular non-autonomous
first-order mechanical systems \cite{LMMMR-2002}, it can be proved that
there are holonomic sections $\psi_\Lag\colon\R\to J^{2k-1}\pi$ of the
projection $\bar\pi^{2k-1}$ such that $\psi_h = \Leg_o \circ\psi_\Lag$,
and they are solutions to the Lagrangian dynamical equations.
Then, the sections $\psi=(\psi_\Lag,\psi_h)$ are solutions to the
Lagrangian-Hamiltonian variational problem (see the proof of
Theorem \ref{prop:LagrangianVariational}).
\end{proof}

In natural coordinates, if $\psi_h$ is given by
$\psi_h(t) = (t,q_i^A(t),p_A^i(t))$, $0 \leqslant i \leqslant k-1$,
then the above equations give the classical higher-order Hamilton equations:
$$
\dot{q}_i^A = \restric{\derpar{H}{p_A^i}}{\psi_h} \quad ; \quad
\dot{p}_A^i = \restric{\derpar{H}{q_i^A}}{\psi_h} \, .
$$

Finally, as a consequence of all the above results, we have the
corresponding relation between vector field solutions to the unified dynamical equations
and those which are solutions to the Hamiltonian equations:

\begin{proposition}
Let $X \in \vf(\W_r)$ be a vector field tangent to $\W_\Lag$ and solution to the equations
\begin{equation}
\inn(X)\Omega_r = 0 \quad ; \quad \inn(X)(\rho_\R^r)^*\eta = 1 \ ,
\label{one}
\end{equation}
Then there exist vector fields $X_h \in \vf(\P)$,
which are solutions to the equations
\begin{equation}
\inn(X_h)\Omega_{\P} = 0 \quad , \quad
\inn(X_h)\bar{\tau}_o^*\eta = 1 \ .
\label{two}
\end{equation}

\noindent Conversely, if $X_h \in \vf(\P)$ is a vector field
which is a solution to the equations \eqref{two},
then there exist vector fields $X \in \vf(\W_r)$ which
are solutions to the equations \eqref{one}.
\end{proposition}
\begin{proof}:
See also \cite{art:Prieto_Roman11_2} (Theorem 4) for a detailed proof of this statement.
\end{proof}

{\bf Remark:}
It is interesting to point out that, for almost-regular systems, if the
unified dynamical equations have consistent solutions on a final constraint submanifold
$\W_f\hookrightarrow \W_r$, then the Lagrangian and Hamiltonian equations
have consistent solutions on final constraint submanifolds
$S_f=\rho_1^r(\W_f)\hookrightarrow J^{2k-1}\pi$
and $\P_f=\hat\rho_1^r(\W_f)\hookrightarrow\P$, respectively.
Then all the above results hold on the points of these submanifolds instead
of $\W_r$, $J^{2k-1}\pi$, and $\P$, respectively.
$$
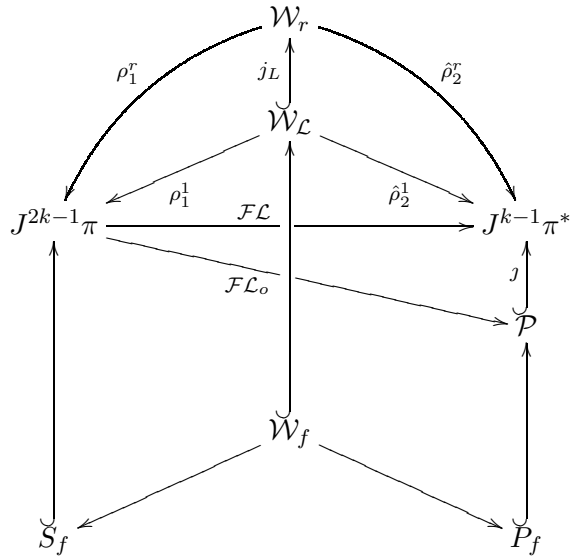
\xymatrix{
\ & \ & \ & \W_r \ar@/_1.3pc/[ddll]_{\rho_1^r} \ar@/^1.3pc/[ddrr]^{\hat{\rho}_2^r} & \ &\  \\
\ & \ & \ & \W_\Lag \ar[dll]^{\rho_1^1} \ar[drr]_{\hat{\rho}_2^1} \ar@{^{(}->}[u]^{j_L} & \ & \ \\
\ & J^{2k-1}\pi \ar[rrrr]^<(0.35){\Leg}|(.493){\hole} \ar[drrrr]_<(0.35){\Leg_o}|(.495){\hole} & \ & \ & \ & J^{k-1}\pi^* \\
\ & \ & \ & \ & \ & \mathcal{P} \ar@{^{(}->}[u]^-{\jmath} \\
\ & \ & \ & \W_f \ar@{^{(}->}[uuu] \ar[dll] \ar[drr] & \ & \ \\
 & S_f \ar@{^{(}->}[uuu] & \ & \ & \ & P_f \ar@{^{(}->}[uu] \\
}
$$

\section{Conclusions and further research}
\label{concl}

We have made an accurate revision of the generalization of the
Lagrangian-Hamiltonian unified formalism of R. Skinner and R. Rusk
to higher-order dynamical systems.
We have analyzed the non-autonomous case,
since the autonomous case can be obtained as a particular situation of this.
This particular situation consists in using trivial bundles and removing the time-dependence
(see \cite{art:Prieto_Roman11_1}).
This unified formalism constitutes a nice framework which
allows us to study different kinds of problems in a simpler way.
In particular, singular (constrained) systems can be
analyzed more easily.

In particular, as a new contribution, the classical variational principles of first-order mechanics
are generalized to this framework, in order to state the
dynamical equations for higher-order mechanics in several equivalent ways.

Therefore,  the Lagrangian and Hamiltonian structures, equations and solutions
of higher-order mechanics are recovered
from those obtained in the unified formalism,
which also includes the corresponding Lagrangian and Hamiltonian variational
principles: the generalized Hamilton and Hamilton-Jacobi Principles
respectively.

Several interesting physical examples have been studied using this formalism;
for instance the Pais-Uhlenbeck oscillator and the shape of a deformed elastic cylindrical beam
with fixed ends, as regular systems;  the second-order relativistic particle,
first as a free particle and later subjected to a potential,
as singular systems \cite{art:Prieto_Roman11_1,art:Prieto_Roman11_2},
and also underactuated control systems \cite{art:Colombo_Martin_Zuccalli10}.

This generalization of the Lagrangian-Hamiltonian unified formalism
to higher-order dynamical systems using a general fibre bundle
$E$ over $\R$ (instead of the classical approach using trivial bundles)
is a first step towards the study of higher-order classical field theory.
However, replacing the base manifold $\R$ with an orientable $m$-dimensional
manifold $M$ gives rise to new difficulties, such as
defining a suitable fiber bundle to act as the phase space; obtaining 
the Legendre map without ambiguities,
or obtaining the relation between the momenta (which is crucial in our formulation).
Nevertheless+, our future aim is to obtain an unambiguous Lagrangian-Hamiltonian
unified formalism for higher-order classical field theory,
thus completing previous works \cite{Campos-et-al,art:Vitagliano10}.

\section*{Acknowledgments}

We acknowledge the financial support of the  {\sl
Ministerio de Ciencia e Innovaci\'on} (Spain), projects
MTM2011-22585 and  MTM2011-15725-E,
and AGAUR, project 2009 SGR:1338..
One of us (PDPM) wants to thank the UPC for a Ph.D grant.
We thank Mr. Jeff Palmer for his assistance in preparing the English
version of the manuscript.


\end{document}